\documentclass[journal]{IEEEtran}
\usepackage{graphicx}
\usepackage{latexsym}
\usepackage{amssymb}
\usepackage{amsmath}
\usepackage{amsthm}
\usepackage{verbatim} 
\newtheorem{theorem}{Theorem}
\newtheorem{lemma}{Lemma}

\newtheorem{definition}{Definition}

\newtheorem{property}{Condition}

\DeclareMathOperator{\Hop}{H}

\newcommand{\Hshannon}{{\Hop}}

\DeclareMathOperator{\emin}{min}

\newcommand{\Hmin}{\Hop_{\emin}}

\DeclareMathOperator{\ext}{ext}
\DeclareMathOperator{\extr}{extr}
\DeclareMathOperator{\conv}{conv}
\DeclareMathOperator{\entropy}{H}

\newcommand{\cancel}[1]{}

\begin{document}

\title{Bit Commitment from Non-Signaling Correlations}

\author{Severin~Winkler, J\"urg~Wullschleger, and Stefan~Wolf
\thanks{S. Winkler and S. Wolf are with the Computer Science Department, ETH Z\"urich, CH-8092 Z\"urich, Switzerland (e-mail: swinkler@ethz.ch; wolf@inf.ethz.ch).}
\thanks{J. Wullschleger is with the Department of Mathematics, University of Bristol, Bristol BS8 1TW, U.K. (e-mail: j.wullschleger@bristol.ac.uk).}}

\maketitle

\begin{abstract}
Central cryptographic functionalities such as encryption, authentication, or secure two-party computation cannot be realized in an information-theoretically secure way from scratch. This serves as a motivation to study what (possibly weak) primitives they {\em can\/} be based on. We consider as such starting points general two-party input-output systems that do not allow for message transmission, and show that they can be used for realizing unconditionally secure bit commitment as soon as they are non-trivial, i.e., cannot be securely realized from distributed randomness only. 
%In particular, our result implies that any two-qubit state  without hidden-variable model has an input-output behavior allowing for unconditional bit commitment. 

\end{abstract}


\IEEEpeerreviewmaketitle

\section{Introduction}
%\subsection{Bit Commitment}

Modern cryptography deals --- besides the classical tasks of encryption and authentication --- with secure cooperation between two (or more) parties willing to collaborate but distrusting each other. Examples of important functionalities of such secure {\em two-party computation\/} are {\em bit commitment\/} and {\em oblivious transfer}. In this note, we concentrate  on bit commitment, a primitive which, for instance, allows for fair coin flipping \cite{Blum82} and has central applications in interactive proof systems. 

A bit commitment scheme is a protocol between two parties, Alice and Bob, that consists of two stages. First, they execute {\em Commit\/} where Alice chooses a bit $b$ as input. Later, they execute {\em Open} where Alice reveals the bit $b$ to Bob. The security properties of bit commitment are the following. Security for Alice ensures that the {\em Commit\/} protocol does not give any information about the bit $b$ to Bob. Security for Bob, on the other hand, means that after the execution of {\em Commit}, $b$ cannot be changed anymore by Alice. Ideally, one would like these security properties to hold even against an adversary with unlimited computing power.

It is well known that unconditionally secure bit commitment cannot be implemented from (noiseless) classical communication only  --- and the same is true even for (noiseless) quantum communication~\cite{Mayers97,LoChau97}. Therefore, it is interesting to study unconditionally secure reductions of bit commitment to weaker primitives, e.g., to physical assumptions. It is known that bit commitment can be realized from communication over noisy channels~\cite{Crepea97}, \cite{WiNaIm03} or from pieces of correlated randomness~\cite{IMNW04},~\cite{WolWul04},~\cite{IMNW06}.
%The present article corrects and strongly extends preliminary results presented in~\cite{WolWul05c}.

Measurements on entangled quantum states can produce so-called {\em non-local correlations}, i.e., correlations that cannot be simulated with shared classical information. These correlations can be modeled as {\em bipartite input-output systems\/} that are characterized by a conditional distribution $P_{XY|UV}$, where $U$ and $V$ stand for the inputs and $X$ and $Y$ for the outputs of the system, respectively. We only consider correlations that are {\em non-signaling}, i.e., which do not allow for message transmission from one side to the other. When using a non-signaling system, a party receives its output immediately after giving its input, independently of whether the other party has given its input already. This prevents the parties from signaling by delaying their inputs. An example of such a system is the {\em non-local box\/} ({\em NL box\/} for short) proposed by Popescu and Rohrlich~\cite{PopRoh94}, where the inputs and outputs are binary, each output is a uniform bit, independent of the pair of inputs, but $X\oplus Y=U\land V$ always holds. 

As bit commitment cannot be implemented from quantum communication, the question has been studied whether bit commitment can be realized when the two parties share trusted non-local correlations as a resource. It has been proven in \cite{BCUWW06} that unconditionally secure bit commitment can be implemented from NL boxes. This result shows that unconditionally secure computation can be realized from non-signaling systems in principle. In particular, it implies that the problems that arise from the fact that any non-signaling system allows the parties to delay their inputs can be circumvented. However, the correlations of an NL box cannot be realized by measurements on a quantum state \cite{Tsi93}. In the present article we show that  {\em any\/} non-signaling system providing binary outputs %allows for unconditionally secure bit commitment. More precisely, we show that any such system 
can either be simulated securely with shared randomness, or allows for information-theoretically secure bit commitment (Theorem \ref{completeness2}); our condition is thus tight. This implies in particular that even local non-signaling correlations can be used to implement unconditionally secure bit commitment if they are provided as a trusted resource to the two parties.

%As bit commitment cannot be implemented from quantum communication, the question has been studied whether bit commitment can be realized when the two parties share such correlations as a resource.

\IEEEPARstart{}{} 

\section{Preliminaries}

\subsection{Bit Commitment}
A bit commitment scheme is a pair of protocols \texttt{Commit} and \texttt{Open} executed by two parties Alice and Bob. First, Alice and Bob execute \texttt{Commit} where Alice has a bit $b$ as input. Bob either accepts or rejects the execution of \texttt{Commit}. Later, they execute \texttt{Open} where Bob has output $(accept,b')$ or $reject$. The two protocols must have the following (ideal) properties:
\begin{itemize}
 \item Correctness: If both parties follow the protocol, then Bob always accepts with $b'=b$.
 \item Hiding: If Alice is honest, then committing to $b$ does not reveal any information about $b$ to Bob.\footnote{Bob's views for $b=0$ and $b=1$ are indistinguishable.}  
 \item Binding: If Bob is honest and accepts after the execution of \texttt{Commit}, then there exists only one value $b'$ (which is equal to $b$, if Alice is honest) that Bob accepts as output after the execution of \texttt{Open}.
\end{itemize}
In the following we call a bit commitment scheme secure, if it fulfills the above ideal requirements except with an error that can be made negligible (as a function of some security parameter $n$).
\subsection{Notation}
Let $W:\mathcal{X}\rightarrow \mathcal{Y}$ be a \textit{stochastic matrix} with rows indexed by elements of $\mathcal{X}$ and columns indexed by elements of $\mathcal{Y}$. We denote the entries of $W$ by $W(y|x)=W_{x}(y)$ and the row vector indexed by $x$ by $W_{x}$. $W_{x}(\cdot)$ defines a probability distribution on $\mathcal{Y}$ for every $x \in \mathcal{X}$, i.e., for all $x$ it holds that $W(y|x)\geq 0$ for all  $y$ and $\sum_{y}W(y|x)=1$.
%end{align*}
We denote by $\conv(W)$ the convex hull of the set $\{W_x|x \in \mathcal{X}\}$, i.e., the convex hull of the row vectors of $W$. We call $W_{x_0}$ an extreme point of this set if the convex hull of the set $(\{W_x|x \in  \mathcal{X}\}\backslash\{W_{x_0}\})$ is strictly smaller. We denote the set of extreme points by $\extr(\conv( W))$. We call $W_{z_{0}}$ non-extreme if it is not an extreme point of $\conv(W)$.
We denote by $x^n=(x_{1},\ldots,x_{n})$ a sequence of elements in $\mathcal{X}$ or a vector in $\mathcal{X}^n$. If $I:=\{i_{1},\ldots,i_{k}\}\subseteq \{1,2,\ldots,n\}$ then $x^I$ denotes the sub-sequence $(x_{i_{1}},x_{i_{2}},\ldots,x_{i_{k}})$ of $x^n$. We denote by $h(\cdot)$ the \textit{binary entropy function}.

We call a function $f(n) \geq 0$ \emph{negligible} if for any $c>0$, there exists $n_c$ such that for all  $n > n_c$, $f(n) < 1/n^c$\;.
We call $f(n)$ \emph{overwhelming} if $1-f(n)$ is negligible.

\subsection{Non-Signaling Boxes}
\label{Boxes}
A \textit{non-signaling box} is defined by a stochastic matrix 
$$W:\mathcal{U} \times \mathcal{V} \rightarrow \mathcal{X} \times \mathcal{Y}$$
as follows:  Alice gives an input $u \in \mathcal{U}$ and Bob gives an input $v \in \mathcal{V}$. Alice gets output $x \in \mathcal{X}$ and Bob $y \in \mathcal{Y}$ with probability $W(xy|uv)$.
Furthermore, the following non-signaling conditions must hold
$$\sum_{y}W(xy|uv)=\sum_{y}W(xy|uv')~~\forall u,v,v',x, $$
$$\sum_{x}W(xy|uv)=\sum_{x}W(xy|u'v)~~\forall u,u',v,y,$$
i.e., the distribution of Alice's output is independent of Bob's input (and vice-versa). A party receives its output immediately after giving its input, independently of whether the other party has given its input already. Note that this is possible, since the box is non-signaling. Furthermore, after a box is used once, it is destroyed. The set of non-signaling boxes can be divided into two types: local and non-local. A box is local if and only if it can be simulated by non-communicating parties with only shared randomness as a resource. This means that there exist probabilities $p_{i}$ and stochastic matrices $V^{i}_{A}, V^{i}_{B}$ such that
\begin{align}
W(xy|uv)=\sum_{i=1}^n p_{i}V^{i}_{A}(x|u)V^{i}_{B}(y|v) ~~\forall u,v,x,y.
\end{align} 
A box is called \textit{independent} if there exist stochastic matrices $V_{A}, V_{B}$ such that 
$$W(xy|uv)=V_{A}(x|u)V_{B}(y|v) ~~\forall u,v,x,y,$$
i.e., such a box can be simulated without any shared resources at all. In the following we only consider boxes with binary outputs, i.e., $\mathcal{X}=\mathcal{Y}=\{0,1\}$. We define
$$W^A(x|u):=\sum_{y}W(xy|uv)~~\forall u,v,x,$$
$$W^B(y|v):=\sum_{x}W(xy|uv)~~\forall  u,v,y.$$
%A box with binary outputs is called \textit{unbiased} if 
%$$W^A(x|u)=W^B(y|v)=1/2~~\forall u,v,x,y$$
We call a box with binary outputs \textit{perfectly correlated} for an input pair $(u,v)\in\mathcal{U}\times\mathcal{V}$ if
$$W(01|uv)=W(10|uv)=0$$
and \textit{perfectly anti-correlated} if
$$W(00|uv)=W(11|uv)=0.$$
An input $u$ for Alice is called \textit{redundant} if there exists $\tilde u\neq u$ such that 
$$W(xy|uv)=W(xy|\tilde uv)~\forall x,y,v.$$
%Without loss of generality we can in the following identify a box with the box obtained after removing the redundant inputs. Thus we assume for the rest of this note that all boxes have no redundant inputs.

%\begin{comment}
\subsection{Chernoff/Hoeffding Bounds}

We will use the following bounds attributed to Chernoff \cite{Cherno52} and Hoeffding \cite{Hoeffd63}.
\begin{lemma}
\label{chernoff1}
Let $X_{1}, X_{2},\ldots,X_{n}$ be independent random variables with $\Pr[X_{i}=1]=p_{i}$ and $\Pr[X_{i}=0]=1-p_{i}$. Let
 $X = \sum_{i=1}^nX_{i}$ and $\mu = E[X]$. Then for any $0 < \delta < 1$ it holds that
\begin{align*}
  \Pr[X > (1+\delta)\mu] &\leq \exp(-\delta^2\mu/3)\;, \\
  \Pr[X < (1-\delta)\mu] &\leq \exp(-\delta^2\mu/2)\;.
\end{align*}
%$\Pr[|X - n \mu|> \delta \mu] \leq 2\exp(-\delta^2\mu/2)$.
\end{lemma}

\begin{lemma}
\label{hoeffding}
Let $X_{1}, X_{2},\ldots,X_{n}$ be independent random variables with $\Pr[X_{i}=1]=p_{i}$ and $\Pr[X_{i}=0]=1-p_{i}$. Let
 $X = \sum_{i=1}^nX_{i}$ and $\mu = E[X]$. Then for any $0 < \delta < 1$ it holds that
\begin{align*}
  \Pr[X > \mu + \delta] &\leq \exp(-2\delta^2/n)\;, \\
 \Pr[X < \mu - \delta] &\leq \exp(-2\delta^2/n)\;.
\end{align*}
%$\Pr[|X - n \mu|> \delta \mu] \leq 2\exp(-\delta^2\mu/2)$.
\end{lemma}

%\end{comment}

\subsection{Information Theory}
We will use the smoothed versions of the min-entropy \cite{RenWol05}. For an event $\mathcal{E}$, let
$P_{X \mathcal{E}|Y=y}(x)$ be the probability that $X=x$
\emph{and} the event $\mathcal{E}$ occurs, conditioned on $Y=y$. We define
\begin{align*}
  \entropy^\epsilon_{\infty}(X|Y) &:= \max_{\mathcal{E}: \Pr(\mathcal E)\geq 1-\epsilon}
  \min_{y} \min_x (-\log P_{X \mathcal{ E}|Y=y}(x)) .
\end{align*}

We will make use of the following lemma from \cite{RenWol05}.
\begin{lemma}
\label{chainRule}
Let $P_{XYZ}$ be a probability distribution. For any $\epsilon,\epsilon'>0$,
$$\entropy_{\infty}^{\epsilon+\epsilon'}(X|YZ)\geq \entropy_{\infty}^{\epsilon}(XY|Z)-\log(|\mathcal Y|)-\log(1/\epsilon')\;.$$
\end{lemma}
The following lemma from \cite{HolRen06} gives a lower bound for the smooth entropy of $n$-fold product distributions:
\begin{lemma}
\label{smoothBound}
 Let $P_{X^nY^n}:=P_{X_{1}Y_{1}}\ldots P_{X_nY_n}$ be a probability distribution over $\mathcal{X}^n\times\mathcal{Y}^n$ and let $\epsilon > 0$. Then
$$\entropy_{\infty}^{\epsilon}(X^n|Y^n)\geq \entropy(X^n|Y^n)-4\sqrt{n\log(1/\epsilon)}\log(|\mathcal{X}|)\;.$$
\end{lemma}

\cancel{
Finally, we will use the following two lemmas about the effect of side information on the min-entropy. 
\begin{lemma}{\em  \textbf{\cite{Wolf99}}}
\label{minEnt1}
 Let P,Q and R be random variables with I(P,R)=0. Then it holds that
$$\mathrm{P}_{QR}[\entropy_{\infty}(P|Q=q,R=r)\geq \entropy_{\infty}(P)-\log|\mathcal{Q}|-s]\geq 1-2^{-s}$$
\end{lemma}
}

\subsection{Randomness Extraction} %and Privacy Amplification}

%In information-theoretic and quantum key agreement, the final protocol step, where a highly secret key is generated from a longer but only weakly secure key, has been called{\rm privacy amplification}. It is  very closely related to randomness extraction; actually, it corresponds to the latter when viewed from a possible  adversary's perspective.
A function $f: \mathcal{X} \times \mathcal{S}\rightarrow\mathcal{Y}$ is called a \textit{two-universal hash function} \cite{CarWeg79} if for all $x_{0}\neq x_{1}$ we have 
$$\Pr[f(x_{0},S)=f(x_{1},S)]\leq \frac{1}{|\mathcal{Y}|}$$
if S is uniform over $\mathcal{S}$.
%A (strong) extractor is a function that takes some initial randomness, called the seed, and a random variable $X$ and outputs 
The following lemma from \cite{BeBrRo88,ILL89} shows that two-universal hash functions are strong extractors, i.e., the concatenation of the seed and the output of the extractor is close to uniform.% i.e., the concatenation of the seedoutput of the hash function 
\cancel{
\begin{definition}{\em \textbf{\cite{CarWeg79}}}
 A function $f: \mathcal{X} \times \mathcal{S}\rightarrow\mathcal{Y}$ is called a \textit{two-universal hash function} if for all $x_{0}\neq x_{1}$ we have 
$$\Pr[f(x_{0},S)=f(x_{1},S)]\leq \frac{1}{|\mathcal{Y}|}$$
if S is uniform over $\mathcal{S}$.
\end{definition}
}

\begin{lemma}[Leftover hash lemma] %\cite{BeBrRo88,ILL89}]
\label{Leftover}
Let $f: \mathcal{X} \times \mathcal{S}\rightarrow\mathcal{Y}$ be a two-universal hash function with $m>0$. Let $X$ be a random variable over $\mathcal{X}$ and let $\epsilon > 0$. If 
$$\entropy_{\infty}(X)-2\log(1/ \epsilon)\geq m,$$
then $\frac12 ||(f(S,X),S) - (U,S)||_1 \leq \epsilon$ for $S$ and $U$ independent and uniform over $\mathcal{S}$ and $\mathcal{Y}$. 
\end{lemma}

\subsection{Typical Sequences}
In this section we will state and prove some basic results on typical sequences. More details on this topic can be found in the book by Csisz\'ar and K\"orner \cite{CsiKor81}.
\begin{definition}
 Let $P$ be a probability distribution on $\mathcal{X}$ and $\epsilon >0$. Then the set of $\epsilon$-\textit{typical sequences} is defined as:
\begin{align*}
 \mathcal{T}_{P,\epsilon}^n:=\{x^n \in \mathcal{X}^n:\forall x \in \mathcal{X}~|N(x|x^n)-P(x)n|\leq \epsilon n \\\text{ and } P(x)=0 \Rightarrow N(x|x^n)=0\},
\end{align*}
where $N(x|x^n)$ denotes the number of letters $x$ in $x^n$. 
\end{definition}

\begin{definition}
 For a stochastic matrix $W:\mathcal{X}\rightarrow \mathcal{Z}$ we define the set of $W$-\textit{typical sequences} under the condition $x^n \in \mathcal{X}^n$ with constant $\epsilon$ as
\begin{align*}
 \mathcal{T}_{W,\epsilon}^n(x^n)=\{z^n:&\forall x,z|N(xz|x^n z^n)-W_{x}(z)N(x|x^n)|\leq \epsilon n 
\\&\text{ and } W_{x}(z)=0 \Rightarrow N(xz|x^n z^n)=0\}.
\end{align*}

\end{definition}
The following two well-known lemmas follow directly from Lemma~\ref{chernoff1}.
\begin{lemma}
\label{typical1}
 $P^n(\mathcal{T}^n_{P,\epsilon})\geq 1-2|\mathcal{X}|\exp(-n\epsilon^2/3)$
\end{lemma}
\begin{comment}
\begin{proof}
 We generate a sequence $x^n=(x_{1},\ldots,x_{n})$ according to $P^n$ and consider the random variables $X_{i}^{y}$ where $y \in \mathcal{X}$ and $1\leq i \leq n$ defined as 
$$X_{i}^y=\begin{cases}
  1,  & \text{ if }x_{k}=y\\
  0,  & otherwise
\end{cases}$$
Then with $\delta_{x}:=\epsilon P(x)^{-1}$ we get
\begin{align*}
&\Pr[x^n \notin \mathcal{T}_{P,\epsilon}]=\Pr[\exists x \in \mathcal{X}:|N(x|x^n)-P(x)n|> \epsilon n]\leq \\ &\sum_{y \in \mathcal{X}}\Pr[|\sum_{i=1}^n
 X_{i}^y-nP(y)| \geq \delta_{y} P(y)n]\leq 2|\mathcal{X}| \exp(-\epsilon^2n/2)
 \end{align*}
by applying the union bound and Lemma \ref{chernoff1}.
\end{proof}
\end{comment}

\begin{lemma}
\label{typical2}
$W_{x^n}^n(\mathcal{T}_{W,\epsilon}^n(x^n))\geq 1-2|\mathcal{X}||\mathcal{Z}|\exp(-n\epsilon^2/3)$
\end{lemma}
\begin{comment}
\begin{proof}
Let $I_{x}:=\{i:x_{i}=x\}$. Using 
\begin{align*}
\mathcal{T}_{W,\epsilon}^n(x^n)&=\{z^n:\forall x,z|N(z|z^{I_{x}})-W_{x}(z)N(x|x^n)|\leq \epsilon n 
\\&\hspace{4cm} \text{ and } W_{x}(z)=0 => N(z|z^{I_{x}})=0\}
\\&=\{z^n:~\forall x~z^{I_{x}} \in \mathcal{T}^{|I_{x}|}_{W_{x},\epsilon P_{x^n}(x)^{-1}}\}
\end{align*}
and the claim follows from Lemma~\ref{typical1}.
\end{proof}
\end{comment}

Using the results above we will prove a lemma that we will use in the security proofs in this paper.
The lemma is similar to Lemma 14 in \cite{WiNaIm03}. Let $W:\mathcal{X}\rightarrow \mathcal{Z}$ be a (discrete memoryless) channel, let $a  \in \mathcal{X}$ be an input such that the output distribution of $a$ is not a convex combination of the other output distributions and let $x^n, \tilde x^n \in \mathcal{X}^n$ be sequences such that  $|\{k:x_k \neq a \text{ and } \tilde x_k=a\}|\geq \kappa n$. Then the lemma states that the output of the channel, given $x^n$ as input, will not be $W$-typical conditioned on $\tilde x^n$ with overwhelming probability if $\exp(-\kappa^2 n)$ is negligible.

\begin{lemma}
\label{statLemma}
Let $W:\mathcal{X}\rightarrow \mathcal{Z}$ be a stochastic matrix and $a \in \mathcal{X}$ such that for all probability distributions $P$ over $\mathcal{X}$ such that $P(a)=0$ and

$$\Big \Vert W_{a}-\sum_x P(x) W_{x} \Big \Vert_{1} \geq \delta\;.$$

Let $x^n, \tilde x^n \in \mathcal{X}^n$ with $d_{H}(x^{I_{a}},\tilde x^{I_{a}})\geq \kappa n$ where $I_{a}:=\{k:\tilde x_{k} = a\}$. If $n_{a}:=|I_{a}|\geq \lambda n$, then

$$W_{x^n}^n(\mathcal{T}^n_{W,\epsilon}(\tilde x^n))\leq 2\exp(-n\epsilon^2/3)$$

\noindent where $\epsilon:=\frac{1}{2|\mathcal{Z}|}\lambda \delta \kappa$.
\end{lemma}

\begin{proof}
Let $D:=\{k \in I_{a}:x_{k} \neq \tilde x_{k}\}$. Then it follows that
\begin{align}
\label{statLemmaInequ}
\nonumber&\Big \Vert\frac{1}{n_{a}}\sum_{k \in I_a}W_{x_{k}}-W_a \Big \Vert_{1}=\frac{|D|}{n_a} \Big \Vert W_{a}-\frac{1}{|D|}\sum_{k \in D}W_{x_{k}} \Big \Vert_{1}\\
&\geq\frac{|D|}{n_a}\delta \geq \kappa\delta\;.
\end{align}
This implies that there exists $b \in \mathcal{Z}$ such that 
$$\Big |\frac{1}{n_{a}}\sum_{k \in I_{a}}W_{x_{k}}(b)-W_a(b) \Big |\geq \frac{1}{|\mathcal{Z}|}\kappa\delta \;.$$
Let $w^n \in \mathcal{T}_{W,\epsilon}^n(\tilde x^n)$. Then it holds that
\begin{align*}
&\Big |N(b | w^{I_{a}})-\sum_{k \in I_{a}}W_{x_{k}}(b) \Big |= \Big |N(ab|\tilde x^n w^n)-\sum _{k \in I_{a}}W_{x_{k}}(b) \Big |
\\&\geq \Big |\sum _{k \in I_{a}}W_{z_{k}}(b)-n_{a}W_{a}(b) \Big |- \Big |n_{a}W_{a}(b)-N(ab|\tilde x^n w^n) \Big |
\\&\geq \frac{1}{|\mathcal{Z}|}\kappa\delta n_{a}-\epsilon n
\\&\geq \frac{1}{2|\mathcal{Z}|}\kappa\delta\lambda n_{a}\;.
\end{align*}
We define independent binary random variables $X_{k},~k \in I_{a}$, with distributions $P_{X_{k}}(1):=W_{x_{k}}(b)$. Let $X=\sum_{k \in I_{a}}X_{i}$ and $\mu:=E[X]=\sum_{k \in I_{a}}W_{x_{k}}(b)$. Let $t:=\frac{1}{2|\mathcal{Z}|}\kappa\delta\lambda n_{a}\mu^{-1}$ (assuming $\mu \neq 0$). Using the Chernoff bound it follows that 
\begin{align*}
 W_{x^n}^n(\mathcal{T}_{W,\epsilon}^n(\tilde x^n))&\leq\Pr \Big [|X-\mu|\geq \frac{1}{2|\mathcal{Z}|}\kappa\delta\lambda n_{a} \Big ]\\
 &=\Pr[|X-\mu|\geq t\mu]\\
 &\leq 2\exp(-\epsilon^2 n/3)\;.
\end{align*}
\end{proof}

\section{Impossibility}
The following theorem proves that a certain class of non-signaling boxes can be securely implemented from shared randomness alone and does, therefore, not allow for unconditinally secure bit commitment (otherwise bit commitment could be implemented form noiseless communication only, which is well known to be impossible). %Informally, an implementation is secure if there exists a protocol for Alice and Bob that perfectly implement the behavior of the box (i.e. the box is local) and Alice and Bob can get the same information as in the protocol (i.e. the shared randomness) if they have access to the box. 
\begin{theorem}
\label{impossibility}
Let a local non-signaling box with binary output be defined by $W:\mathcal{U}\times\mathcal{V} \rightarrow \{0,1\}^2$ such that 
$$W(xy|uv)=pV_{A}^0(x|u)V_{B}^0(y|v) + (1-p)V_{A}^1(x|u)V_{B}^1(y|v)$$ and there exists $u_{0}\in\mathcal{U},v_{0}\in\mathcal{V}$ and $b_0,b_1\in\{0,1\}$ with:
\begin{align*}
V_A^0(0|u_{0})&=V_A^1(1|u_{0})=b_{0}\\
V_A^0(1|u_{0})&=V_A^1(0|u_{0})=1-b_{0}\\
V_B^0(0|v_0)&=V_B^1(1|v_0)=b_1\\
V_B^1(1|v_0)&=V_B^0(0|v_0)=1-b_1\;,
\end{align*}
then there is no reduction of information-theoretically secure bit commitment to the box $W$ (with noiseless communication only).
\end{theorem}

\begin{proof}
 We prove the statement by showing that one can securely implement such a box from noiseless communication and shared randomness alone. %This would allow for bit commitment from noiseless communication which is impossible as mentioned above. 
The implementation directly follows the definition of the box: Let $\lambda$ be the shared random bit. Alice on input $u$ outputs $0$ with probability $V^{\lambda}_A(0|u)$ and $1$ with probability $V^{\lambda}_A(1|u)=1-V^{\lambda}_A(0|u)$. Bob on input $v$ outputs $b\in\{0,1\}$ with probability $V^{\lambda}_B(b|v)$. This perfectly implements the behavior of the box. Furthermore, this implementation is secure, since Alice and Bob can get the same information (i.e. the shared randomness $\lambda$) if they only have black-box access to $W$, if they always input $u_{0}$ and $v_{0}$, respectively.
\end{proof}

\cancel{
\begin{proof}
 We prove the statement by showing that one can securely implement such a box from noiseless communication and shared randomness alone. This would allow for bit commitment from noiseless communication which is impossible as mentioned above.  By 'securely' implement we mean that we can give protocols for Alice and Bob that perfectly implement the behavior of the box (i.e. the box is local). Additionally, Alice and Bob can get the same information (i.e. the shared randomness) if they have access to the box. Formally, there exists a simulator (applied to the box) for Alice (Bob) that simulates exactly what an adversary would get in the execution of the protocol (for a more detailed and more formal definition we refer to the chapter 'Secure Two-Party Computation' in \cite{Wullsc07b}). The implementation just follows the definition of the box: Let $\lambda$ be the shared random bit. Alice on input $u$ outputs $0$ with probability $V^{\lambda}_A(0|u)$ and $1$ with probability $V^{\lambda}_A(1|u)=1-V^{\lambda}_A(0|u)$.
%the shared random bit $\lambda$ and with probability $V(1-\lambda|u)=1-V(\lambda|u)$ the flipped bit $1-\lambda$. 
Bob on input $v$ outputs $b\in\{0,1\}$ with probability $V^{\lambda}_B(b|v)$. This perfectly implements the behavior of the box. A simulator for Alice always inputs $u_{0}$ (getting the shared randomness). A simulator for Bob always inputs $v_{0}$.
\end{proof}
}

\section{Two Protocols}

We will now give two slightly different protocols, which work for two different kinds of non-signaling boxes.

\subsection{Protocol I}
\label{protocol1}
Informally, the first protocol works as follows: in the \texttt{Commit} protocol an honest Alice gives a fixed input to all her boxes, while Bob chooses his inputs randomly. Alice applies privacy amplification to the outputs of the boxes and uses the resulting key $K$ to hide the bit $B$ she wants to commit to. Alice then sends $K\oplus B$ and the randomness used for privacy amplification to Bob. In the \texttt{Open} protocol Alice sends her outputs from the boxes. Alice's input is chosen such that there is a statistical test  that allows Bob to detect if Alice has changed more than $O(\sqrt{n})$ output values while Bob has only limited information about the output of the boxes before the opening phase. A dishonest Alice might still be able to change $O(\sqrt{n})$ output values. To ensure that this is not possible, we use a linear code and let Alice  send parity check bits of the output to Bob in the \texttt{Commit} protocol. If the minimal distance of the code is large enough, no two strings with the same parity check bits lie in a hamming sphere with radius proportional to $\sqrt{n}$.
%A solution would be to use a linear code and let Alice send parity check bits of the output to Bob in the \texttt{Commit} protocol. If the minimal distance of the code is large enough, no two strings with the same parity check bits lie in a hamming sphere with radius proportional to $\sqrt{n}$. Instead of a linear code we use 2-universal hashing in the protocol below: Alice's cheating probability is then negligible in the length of the hash value if she changes only $O(\sqrt{n})$ values. 

Let Alice and Bob share $n$ identical non-signaling boxes given by $W: \mathcal{U}\times\mathcal{V} \rightarrow \{0,1\}^2$. In our protocol, we will require Bob to choose his input uniformly from $\mathcal{V}$. For an honest Bob and a potentially malicious Alice, we can define a stochastic matrix
%$W^\prime: \mathcal{U}\rightarrow \{0,1\}^2\times\mathcal{V}$ describing the probability of
%Bob choosing input $v$ and getting output $y$ and Alice getting output $x$ conditioned on Alice's input $u$ as
%$$W^\prime(xyv|u):=\frac{1}{|\mathcal{V}|}W(xy|uv)$$
%From $W^\prime$ we can define the stochastic matrix 
$\hat W: \{0,1\}\times \mathcal{U} \rightarrow \{0,1\}\times\mathcal{V}$ describing the probability of Bob's input and output values $v$ and $y$, %choosing input $v$ and getting output $y$ 
conditioned on Alice's input $u$ and output $x$ as 
$$\hat W(yv|xu):= %\frac{W'(xyv|u)}{W^A(x|u)}=
\frac{1}{|\mathcal{V}|}\frac{W(xy|uv)}{W^A(x|u)}\;,$$
if $W^A(x|u) \neq 0$, and undefined otherwise. %else $\hat W(yv|xu):=0$.
%\Comment{Ich finde = 0 macht keinen sinn, denn dann ist ja fuer $W^A(x|u) = 0$ $W(yv|xu)$ gar keine wahrscheinlichkeitsverteilung mehr.}
Furthermore, we will require an honest Alice to always input a fixed value $u_{a}$ to the box. For an honest Alice, and a potentially malicious Bob that chooses his input  $v \in \{0,1\}$ freely, we can define random variables $X_{v},Y_{v}$ depending on Bob's input that describe the output of Alice and Bob, respectively, i.e. with a joint distribution 
$$P_{X_{v}Y_{v}}(x,y):=W(xy|u_{a}v).$$
The protocol below is secure if there exists a value $a=(x_a,u_a)$ such that the following condition is fulfilled:  
\begin{property}
\label{property1}
(1) There exists $\delta>0$ such that for all probability distributions $P$ over $\{0,1\}^2$ with $P(a)=0$ it holds that 
$$\Big \Vert \hat W_{a}-\sum_{x}P(x) \hat W_{x} \Big \Vert_{1} \geq \delta\;.$$
(2) There exists $\gamma > 0$ such that for all $v\in \mathcal{V}$ it holds that 
$$\Hshannon(X_{v}|Y_{v})\geq \gamma,$$
i.e., the Shannon entropy of Alice's output given Bob's output is non-zero for all possible inputs of Bob.
\end{property}
We label the inputs of Alice as $\{0,\ldots,|\mathcal{U}|-1\}$.
% \Comment{Ich finde es klarer, wenn anstatt $a=(x_a,u_a)=(0,0)$ einfach $x_a$ und $u_a$ verwendet werden.} %In the following we assume that $a=(x_a,u_a)=(0,0)$.
Furthermore, we define the distribution of Alice's output $x$ if her input is $u_a$ as $P(x):=W^A(x|u_a)$ for all $x\in\{0,1\}$. 
\begin{comment}
Let $\lambda:=\frac{1}{2}\text{min}\{P(x),~x \in \{0,1\}\}$. We choose $k:=n^{2/3}$ which implies that $k, k \log(n), k/n, \sqrt{n k} \in o(n)$ and $k,k^2,k^2/n \in \omega(1)$. Let $\epsilon:=\frac{1}{4} \lambda \delta k /n$, $m := \log(k) + k \log(3n) + k$
%Let $d\geq2k_1+1$. Let $R$ such that $l:=\gamma n-n(1-R)-4k_1-3k_2>0$ and such that a binary linear $[n,Rn,d]$-code exists (e.g. a concatenated code \cite{Forney66},\cite{Justes72}). Let $H$ be the parity check matrix of this code.
and $l:=\gamma n-m-4\sqrt{n k}-3k > 0$. (Note $l \in \gamma n - o(n)$, so if $n$ is big enough, we have $l >0$.)
 Let $f:\{0,1\}^*\times \{0,1\}^n\rightarrow \{0,1\}^m$ and $\ext:\{0,1\}^*\times \{0,1\}^n\rightarrow \{0,1\}^l$ be  two-universal hash functions.
\end{comment}
Let $\lambda:=\frac{1}{2}\text{min}\{P(x),~x \in \{0,1\}\}$. Let $k$ be the security parameter. Let $\epsilon:=\frac{1}{4} \lambda \delta k/n$. Let $d>2k$ and let $H$ be the parity check matrix of a linear $[n,Rn,d]$-code with $R>(1-\gamma)$. Since we do not have to decode, this could be a random linear code chosen by Bob. Let  $l:=\gamma n-n(1-R)-4\sqrt{nk}-3k$. We choose $k:=n^{2/3}$, which implies that $k, \sqrt{nk} \in O(n^{5/6})$ and $k, k^2/n, n \epsilon^2 \in \Omega(n^{1/3})$. It follows that $l\in (\gamma + R - 1)n-O(n^{5/6})$. If $n$ is big enough, we have $l >0$. Let $\ext:\{0,1\}^*\times \{0,1\}^n\rightarrow \{0,1\}^l$ be a two-universal hash function. We define $syn(x^n):=H^\mathrm{T}x^n$.

\vspace{0.2cm}
 %\fbox{%
  %\begin{minipage}[b]{140mm}

%\fbox{
\noindent
Commit($b^l$):
\begin{itemize}
 \item Bob chooses $v^n \in_{R} \{0,1\}^n$
 \item Alice and Bob input $u_a^n$ and $v^n$ component-wise to the boxes. Alice gets $x^n \in \{0,1\}^n$ and Bob $y^n \in \{0,1\}^n$.
 \item Alice chooses $r \in_{R} \{0,1\}^*$ and sends \\$(syn(x^n),r,b^l \oplus \ext(r,x^n))$ to Bob.
\end{itemize}
Open():
\begin{itemize}
\item Alice sends Bob $x^n \text{ and } b^l$.
\item Bob checks:
\begin{itemize}
\item $syn(x^n)$ is correct
\item $b \oplus \ext(r,x^n)$ is correct
\item $((y_{1},v_{1}),..,(y_{n},v_{n}))\in\mathcal{T}^n_{\hat W,\epsilon}((x_{1},u_a),..,(x_{n},u_a))$
\item $x^n \in \mathcal{T}^n_{P,\epsilon}$
\end{itemize}
\item If all the checks pass successfully, Bob accepts and outputs $b^l$, otherwise he rejects. 
\end{itemize}
%\end{minipage}}

\subsection{Security}
Let $u^n:=(u_{1},\ldots,u_{n})$ be Alice's inputs to the boxes, let $x^n:=(x_{1},\ldots,x_{n})$ be her outputs from the boxes and let $\tilde x^n:=(\tilde x_{1},\ldots,\tilde x_{n})$ be the values Alice sends to Bob in the opening phase. We define $z^n:=((x_{1},u_{1}),\ldots,(x_{n},u_{n}))$ and $\tilde z^n:=((\tilde x_{1},u_a),\ldots,(\tilde x_{n},u_a))$. Let $r^n:=((y_{1},v_{1}),\ldots,(y_{n},v_{n}))$ be Bob's inputs and outputs. 
\begin{lemma}
The protocols \texttt{Commit} and \texttt{Open} satisfy the correctness condition.
\end{lemma}
\begin{proof}
Bob always accepts \texttt{Commit}. If Alice follows the protocol, then $syn(x^n)$ and $b^l \oplus \ext(r_{1},u^n)$ are correct. From Lemma~\ref{typical2} it follows that 
\begin{align*}
\Pr[r^n \in \mathcal{T}_{\hat W,\epsilon}^n(z^n)]&=\hat W_{z^n}(\mathcal{T}_{\hat W,\epsilon}(z^n))\\&\geq 1-16|\mathcal V|\exp(-n\epsilon^2/3)\;,
\end{align*}
and from Lemma~\ref{typical1} it follows that 
\begin{align*}
\Pr[x^n \in \mathcal{T}_{P,\epsilon}^n]
&=P^n_{x^n}(\mathcal{T}_{P,\epsilon})\\
&\geq 1-4\exp(-n\epsilon^2/3)\;.
\end{align*}
Thus, Bob accepts \texttt{Open} with overwhelming probability and outputs $b^l$, the value Alice was committed to.
\end{proof}

\begin{lemma}
 The protocol \texttt{Commit} satisfies the privacy condition with an error negligible in $n$.
\end{lemma}
\begin{proof}
Let us assume that Alice is honest. Alice inputs $u_a$ into the boxes as required by the protocol, while Bob can choose its input $v^n=(v_1,\ldots,v_n)$ freely. We then define the random variables $X^n=X_{v_{1}}\times\ldots\times X_{v_{n}}$ and $Y^n=Y_{v_{1}}\times\ldots\times Y_{v_{n}}$. 
Let $\epsilon_{1}:=2^{-k}$.
%Define $\epsilon_{1}:=2^{-k}$ and $\epsilon_{2}:=2^{-ck}$, where $c:=2log^2(2+3)$. 
According to Lemma~\ref{smoothBound} it holds that
$$\Hmin^{\epsilon_{1}}(X^n|Y^n)\geq \entropy(X^n|Y^n)-4\sqrt{n k}.$$
Using Lemma~\ref{chainRule} with get that
\begin{align*}
\Hmin^{2\epsilon_{1}}(X^n|syn(X^n)V^n)&\geq
\Hmin^{\epsilon_{1}}(X^n|Y^n)-n(1-R)\\&~~-\log(1/\epsilon_{1})\\
&\geq \gamma n-n(1-R)-4\sqrt{n k}-k\\\nonumber&= l+2k\;.
\end{align*}
According to Lemma~\ref{Leftover} Bob has no information about $\ext(r_{1},x^n)$ except with probability $2\epsilon_1+\epsilon_1$.
\end{proof}
%Obviously, delaying her Inputs doesn't help Alice.
%\Comment{sollen wir noch erwaehnen, dass es hier egal ist wann Alice ihren input gibt?}

\begin{lemma}
\label{statLemma3}
If $d_{H}(x^n,\tilde x^n)\geq k$, then the probability that Bob accepts $\tilde x^n$ is negligible in $n$.
\end{lemma}
\begin{proof}
From $d_{H}(x^n,\tilde x^n)\geq k$ follows $d_{H}(z^n,\tilde z^n)\geq k$. Let $n_{a}:=N(u_a|u^n)$,  $I_{a}:=\{k:~\tilde z_{k}=(x_a,u_a)\}$ and $p := W^A(x_a|u_a)$.
For all $w^n \in \mathcal{T}^n_{P,\epsilon}$, we have
$$|N(x_a|w^n)-np|\leq \epsilon n = \frac 1 4 \lambda \delta k/n \cdot n \leq \frac{1}{8}k p\;,$$
since $\lambda \leq p / 2$ and $\delta \leq 1$.
We distinguish two cases:

\noindent
(1) $n_{a}\leq (n-k/2)$:   The expectation of $N((x_a,u_a)|z^n)$ is smaller than or equal to $(n-\frac{k}{2})p$.
Since $k^2/n\in \Omega(n^{1/3})$, it follows from Lemma \ref{chernoff1} that with overwhelming probability %Thus, it follows from Lemma \ref{chernoff1} that with probability $1 - \exp(-(k/8)^2 p/n)$,
%$\delta\mu = kp/8 \rightarrow \delta=\mu^{-1}kp/8 \rightarrow \text{ prob. } \leq \exp(-\delta^2\mu/3) \approx k^2/n$
\begin{align*}
N((x_a,u_a)|z^n)
&\leq \left(n-\frac{k}{2}\right)p+\frac{k }{8} p\\
&=\left(n-\frac{3}{8}k\right) p.
\end{align*}
But since Bob only accepts if $\tilde x^n \in \mathcal{T}^n_{P,\epsilon}$, we have
\begin{align*}
d_{H}(z^{I_{a}},\tilde z^{I_{a}})
&\geq \left(n-\frac{1}{4}k \right)p - \left(n-\frac{3}{8}k  \right)p \\
&= \frac{1}{4}k p
\end{align*}
and the claim follows from Lemma~\ref{statLemma}.

\noindent
(2) $n_{a}>(n-k/2)$: Then the expectation of $N((1 - x_a,u_a)|z^n)$ is greater than or equal to $(n - \frac{k}{2})(1-p)$. As $k^2/n\in \Omega(n^{1/3})$ Lemma \ref{chernoff1} implies that with overwhelming probability 
\begin{align*}
 N(( 1 -x_a,u_a)|z^n)
&\geq \left(n- \frac{k}{2}\right)(1-p)-\frac{k}{8}(1-p)\\
&=n(1-p)-\frac{5}{8}k (1-p).
\end{align*}
But since Bob only accepts if $\tilde x^n \in \mathcal{T}^n_{P,\epsilon}$, we have
\begin{align*}
d_{H}(z^{I_{a}},\tilde z^{I_{a}})
&\geq \left(n-\frac{5}{8}k\right)(1-p) - \left(n-\frac{1}{4}k\right)(1-p) \\
&= \frac{1}{4}k (1-p)
\end{align*}
and the claim follows from Lemma~\ref{statLemma}.
\begin{comment}
 $d_{H}(z^{I_{0}},\tilde z^{I_{0}})< \frac{1}{2^3}k$, then 
\begin{align*}
N(0|\tilde x^n)
&\geq N(10|z^n)+(k_1-d_{H}(z^I_{0},\tilde z^I_{0})) \\
&\geq nW^A(1|0)+\frac{1}{2^2}k_1 W^A(1|0)
\end{align*}
This implies $\tilde x^n \notin \mathcal{T}^n_{P,\epsilon}$. If $d_{H}(z^{I_{0}},\tilde z^{I_{0}})\geq \frac{1}{2^3}k$, then the claim follows from Lemma~\ref{statLemma}.
\end{comment}
\end{proof}

\begin{lemma}
The protocol satisfies the binding condition with an error negligible in $n$.
\end{lemma}
\begin{proof}
Any two strings $s^n \neq \tilde s^{n}$ with $syn(s^n)=syn(\tilde s^{n})$ have distance at least $d$. So at least one of the two strings has distance at least $k$ from Alice's output $x^n$. The probability that Bob accepts this string in the opening phase is negligible according to lemma \ref{statLemma3}.
\end{proof}

%\newpage
\subsection{Protocol II} 
Protocol I is not hiding if for every fixed input of Alice a dishonest Bob can choose an input such that he has perfect information about Alice's output. This is the case for example with the above mentioned {\em NL box\/}. But, as shown in \cite{BCUWW06}, this box allows for bit commitment. Therefore, we present a second protocol that allows to securely implement bit commitment for such boxes. The protocol, which is similar to a protocol proposed without a security proof in \cite{WolWul05c} already, works as follows: Alice gives random inputs to all her boxes. Then she applies privacy amplification to the string of inputs and uses the resulting key to hide the bit she is committed to. In the opening phase Alice sends all her inputs/outputs to Bob. Bob performs statistical tests on the input/output of Alice that allow him to detect if Alice has changed more than $\sqrt{n}$ values. We use again parity check bits of a linear code to make sure that a dishonest Alice cannot change $\sqrt{n}$ values except with negligible probability.

 Alice and Bob share $n$ identical non-signaling boxes given by $W: \mathcal{U}\times\mathcal{V}\rightarrow \{0,1\}^2$. We define the corresponding matrix $\hat W$ as in Section~\ref{protocol1}. In the following we always assume that $W^A(x|u)\neq0$ for all $x\in \{0,1\},u\in\mathcal{U}$. For the following protocol to be secure we require $W$ to fulfill the following condition:
\begin{property}
\label{non-trivial3}
There exist $u_{0},u_{1} \in \mathcal{U},~u_{0}\neq u_{1}$, such that the set $D:=\{\hat W_{0u_{0}},\hat W_{1u_{1}},\hat W_{0u_{0}},\hat W_{1u_{1}}\}$ contains at most one non-extreme point of $\conv(\hat W)$, i.e., there is $c_{0} \in \{0u_{0},1u_{1},0u_{0},1u_{1}\}$ such that for all $c\in \{0u_{0},1u_{1},0u_{0},1u_{1}\}\backslash \{c_{0}\}$ it holds that for all probability distributions $P$ with $P(c)=0$
\begin{align*}
\Big \Vert \hat W_{c}-\sum_{z}P(z) \hat W_{z} \Big \Vert_{1} \geq \delta\;.
\end{align*}
\end{property}

We label the inputs of Alice as $\{0,\ldots,|\mathcal{U}|-1\}$ and assume that $u_{0}=0$ and $u_{1}=1$. In the protocol, we will require Alice to choose her input uniformly from $\{0,1\}$, and Bob to choose his input uniformly from $\mathcal V$. If both are honest, the joint distribution of the inputs and outputs of Alice and Bob is
$$P(x,y,u,v):=
\begin{cases}
 \frac{1}{2|\mathcal{V}|}W(xy|uv),  & \text{if }~u \in \{0,1\}\\
 0, &  \text{else}.
\end{cases}
$$
%$$P(x,y,u,v):= \frac{1}{2|\mathcal{V}|}W(xy|uv),~u \in \{u_{0},u_{1}\},x\in\{0,1\}$$
If Alice is honest, the joint distribution of her input and output is
$$Q(x,u):=
\begin{cases}
 \frac{1}{2}W^{A}(x|u),  & \text{if }~u \in \{0,1\}\\
 0, &  \text{else.}
\end{cases}
$$

%Note that $\sum_{v,y} P(x,y,u,v)=\frac{1}{2}W^{A}(x|u)$, if $u \in \{0,1\}$.

%$$P(x,y,u,v):= \frac{1}{2|\mathcal{V}|}W(xy|uv),~u\in\{u_0,u_1\},x\in\{0,1\}$$
%and a joint distribution of Alice's input and output
%$$Q(x,u):= \sum_{v,y} P(x,y,u,v)=\frac{1}{2}W^{A}(x|u),~u\in\{u_0,u_1\},x\in\{0,1\}$$
Let $\lambda:=\frac{1}{4}\text{min}\{Q(x,u), (x,u) \in \{0,1\}^2\}.$
Let $p_{0}:=\min\{W^A(x|u),(x,u) \in \{0,1\}^2\}$. Note that we assumed $p_0 > 0$ and that obviously we also have $p_0 \leq \frac12$. Let $k_1$ be the security parameter, $k_{2}:=k_1(4p_{0}+1)/2p_{0}^2$, $\epsilon:=\frac{1}{4} \lambda \delta k_1/n$, $d\geq k_1+2k_2+1$, $l>0$ and let $H$ be the parity check matrix of a $[n,Rn,d]$-linear code with $Rn\geq n/2+\frac{3}{2}k_1+l/2$. We choose $k_1:=n^{2/3}$ and $l:=n-2n(1-R)-3k_1$. This implies $k_1,k_1^2/n,n\epsilon^2 \in \Omega(n^{1/3})$ and $l \in (2R-1)n-O(n^{2/3})$.  If $n$ is big enough, then $l>0$. Let $\ext:\{0,1\}^*\times \{0,1\}^n\rightarrow \{0,1\}^l$ be a two-universal hash function.
%$m_1 := \log(k_1) + k_1 \log(3 n) + k_1$, $m_2 := \log(k_1 + k_2) + (k_1 + k_2) \log(3 n) + k_1$, and $l := n - m_1 - m_2- 3k_1$. We choose $k_{1}:=n^{2/3}$, which implies that $m_1,m_2 \in O(n^{2/3} \log(n))$ and $k_1, k_{1}^2/n, n \epsilon^2 \in \Omega(n^{1/3})$. It follows that $l \in \gamma n - O(n^{2/3}\log(n))$, so if $n$ is big enough, we have $l >0$. Let $f_1:\{0,1\}^*\times \{0,1\}^n \rightarrow \{0,1\}^{m_1}$, $f_2:\{0,1\}^*\times \{0,1\}^n\rightarrow \{0,1\}^{m_2}$ and $\ext:\{0,1\}^*\times \{0,1\}^n\rightarrow \{0,1\}^l$ be 2-universal hash functions. Let $b^l \in \{0,1\}^l$.

\begin{comment}
Let $\lambda:=\frac{1}{4}\text{min}\{Q(x,u), (x,u) \in \{0,1\}^2\}.$ Let $p_{0}:=\min\{W^A(x|u),(x,u) \in \{0,1\}^2\}$. Note that we assumed $p_0 > 0$. Obviously we also have $p_0 \leq \frac12$. Let $k_{1}:=n^{2/3}$ (which implies that $k_{1},k_{1}^2/n\in \omega(1)$), $k_{2}:=k_1(4p_{0}+1)/2p_{0}^2$,
$m_1 := \log(k_1) + k_1 \log(3 n) + k_1$,
$m_2 := \log(k_1 + k_2) + (k_1 + k_2) \log(3 n) + k_1$,
and $l := n - m_1 - m_2- 3k_1>0$. (Note that $l \in n - o(n)$, so if $n$ is big enough, we have $l >0$.)
Let $\epsilon:=\frac{1}{4} \lambda \delta k_1/n$.
Let $f_1:\{0,1\}^*\times \{0,1\}^n \rightarrow \{0,1\}^{m_1}$, $f_2:\{0,1\}^*\times \{0,1\}^n\rightarrow \{0,1\}^{m_2}$ and $\ext:\{0,1\}^*\times \{0,1\}^n\rightarrow \{0,1\}^l$ be two-universal hash functions. Let $b^l \in \{0,1\}^l$.
\end{comment}

\vspace{0.2cm} 
\noindent
Commit($b^l$):
\begin{itemize}
 \item Alice chooses $u^n \in_{R} \{0,1\}^n$, Bob chooses $v^n \in_{R} \mathcal{V}$.
 \item Alice and Bob input $u^n$ and $v^n$ component-wise to the boxes. Alice gets $x^n \in \{0,1\}^n$ and Bob $y^n \in \{0,1\}^n$.
 %\item Alice sends Bob $syn(u^n),syn(x^n)$.
 \item Alice chooses $r_{2} \in_{R} \{0,1\}^*$ and sends \\$(syn(u^n),syn(x^n),r_{2},b^l \oplus \ext(r_{2},x^n))$ to Bob.
\end{itemize}
Open():
\begin{itemize}
\item Alice sends Bob $u^n,x^n \text{ and } b^l$.
\item Bob checks:
\begin{itemize}
%\item $s_{0}= syn(u^n)$ 
\item $syn(u^n)$ and $syn(x^n)$  are correct
\item $b^l \oplus \ext(r_{2},u^n)$ is correct
\item $((y_{1},v_{1}),..,(y_{n},v_{n})) \in \mathcal{T}_{\hat W,\epsilon}^n((x_{1},u_{1}),..,(x_{n},u_{n}))$
\item $((x_{1},u_{1}),\ldots,(x_{n},u_{n})) \in \mathcal{T}_{Q,\epsilon}^n$
\end{itemize}
\item If all the checks pass successfully, Bob accepts and outputs $b^l$, otherwise he rejects. 
\end{itemize}

\subsection{Security}
Let $z^n:=((x_{1},u_{1}),\ldots,(x_{n},u_{n}))$ be Alice's input and output, $\tilde z^n:=((\tilde x_{1},\tilde u_{1}),\ldots,(\tilde x_{n},\tilde u_{n}))$ the values Alice sends to Bob in the opening phase and $r^n:=((y_{1},v_{1}),\ldots,(y_{n},v_{n}))$ Bob's inputs and outputs. For all $c \in (\{0,1\}\times \mathcal{U})$ we define the sets $I_{c}:=\{i:\tilde z_{i}=c\}$.
\begin{lemma}
The protocols \texttt{Commit} and \texttt{Open} satisfy the correctness condition.
\end{lemma}
\begin{proof}
Bob always accepts \texttt{Commit}. If Alice follows the protocol, then $syn(u^n), syn(x^n)$ and $b^l \oplus \ext(r_{2},u^n)$ are correct. From Lemma~\ref{typical2} it follows that 
\begin{align*}
\Pr[r^n \in \mathcal{T}_{\hat W,\epsilon}^n(z^n))]&=\hat W_{z^n}(\mathcal{T}_{\hat W,\epsilon}(z^n))\\&\geq 1-8|\mathcal U||\mathcal V|\exp(-n\epsilon^2/2)
\end{align*}
and from Lemma~\ref{typical1} it follows that 
\begin{align*}
\Pr[z^n \in \mathcal{T}_{Q,\epsilon}^n]
&=Q^n(\mathcal{T}_{Q,\epsilon}(z^n))\\
&\geq 1-4|\mathcal U|\exp(-n\epsilon^2/2)\;.
\end{align*}
Thus, Bob accepts \texttt{Open} with overwhelming probability and outputs $b^l$, the value Alice was committed to.
\end{proof}

\begin{lemma}
 The protocol \texttt{Commit} satisfies the privacy condition with an error negligible in $n$.
\end{lemma}

\begin{proof}
Let us assume that Alice is honest. Since the box is non-signaling, Bob's values $Y^n$ and $V^n$ are independent of $U^n$. Since Alice chooses $U^n$ uniformly from $\{0,1\}^n$, we have
\[\entropy_\infty(U^n) = n\;.\]
%The additional randomness $R_{2}$ chosen by Alice is independent of $U^n$, so 
All the information Bob gets about $U^n$ is $syn(U^n)$ and $syn(X^n)$. Let $\epsilon_1 := 2^{-k_1}$. Using Lemma~\ref{chainRule} we get
\begin{align*}
\entropy^{\epsilon_1}_\infty(U^n|syn(U^n)syn(X^n)) &\geq n - 2n(1-R) - k_1\\&\geq l+2k_1\;.
\end{align*}
%that Bob's min-entropy about $u^n$ is at least $k-m_1-m_2-s$ with probability
%at least $1-2^{-k_1}$.
If follows from Lemma~\ref{Leftover} that extracting $l$ bits makes the key uniform with an error of at most $2\epsilon_1 = 2 \cdot 2^{-k_1}$. The statement follows.
\end{proof}

The proof of the binding condition is slightly more involved. Because our boxes are non-signaling, Alice has the possibility of delaying her input to the box until the opening phase. Hence, a general strategy for her is to
give input to some of the boxes in the commit phase, and to delay the input to some of the boxes until the opening phase. And she may send incorrect values about her input/output to/from the boxes to Bob in the opening phase.
Note that we can ignore the case where she does not give any input to some boxes, as she might as well just give input but ignore the output.

\begin{lemma}
\label{statLemma4}
If $d_{H}(z^n,\tilde z^n)\geq k_{1}$, then the probability that Bob accepts $\tilde z^n$ is negligible.
\end{lemma}

\begin{proof}
For all $w^n \in \mathcal{T}_{Q,\epsilon}^n$ it holds that
\[|N(xu|w^n)-n Q(x,u)|\leq \epsilon n = \frac 1 4 \lambda \delta k_1/n \cdot n \leq \frac{1}{64} k_1\;, \]
since $\lambda \leq \min_{x,u} Q(x,u) / 4 \leq 1/16$ 
%\Comment{Wuerde nicht auch $\lambda \leq 1/16$ gelten? Denn $\min_{x,u} Q(x,u) \leq 1/4$.} 
and $\delta \leq 1$.

We distinguish the following two cases:\\
(1) There exists $u' \in \{0,1\}$ 
%\Comment{Vielleicht ist $u'$ auch kein guter name, aber ich wollte einfach, dass man erkennt, ob es input oder output ist. dann liest es sich leichter.} 
such that $N(u'|u^n)\leq n/2-k_1/8$:
For all $x \in \{0,1\}$ the expectation of
$N(x u'|z^n)$ is equal to
$(n/2-\frac{k_1}{8}) W^A(x|u')$.
%Thus, it follows from Lemma \ref{chernoff1} that with probability $1 - \exp(-n^{3/2}/16^2) W^A(x|u'))$,
Since $k_{1}^2/n\in \Omega(n^{1/3})$ it follows from Lemma \ref{chernoff1} that with overwhelming probability
\begin{align*}
N(x u'|z^n)
&\leq \Big (\frac{n}{2}
-\frac{1}{8}k_1 \Big)  W^A(x|u')
+\frac{1}{16}k_1 W^A(x|u')\\
&= \Big (\frac{n}{2}-\frac{1}{16}k_1 \Big) W^A(x|u')\;.
\end{align*}
But since Bob only accepts if $\tilde z^n \in \mathcal{T}_{Q,\epsilon}^n$, we have
$d_{H}(z^{I_{0u'}},\tilde z^{I_{0u'}})\geq \frac{1}{32}k_1$ and $d_{H}(z^{I_{1u'}},\tilde z^{I_{1u'}})\geq \frac{1}{32}k_1$, 
and the claim follows from Lemma~\ref{statLemma}.\\

(2) For all $u \in \{0,1\}$ we have $|n/2-N(u|u^n)|\leq k_1/8$: Since $\epsilon^2 n\in \Omega(n^{1/3})$ it follows from Lemma~\ref{typical2} that with overwhelming probability 
%\Comment{Noch zu berechnen. Sollte man aber machen, um zu zeigen, dass z.b. $k_1 = \sqrt{n}$ eben zu klein ist.},
 we have  $z^n\in \mathcal{T}_{W^A,\epsilon}^n(u^n)$. 
Assume 
$z^n\in \mathcal{T}_{W^A,\epsilon}^n(u^n)$.
There exists a value $(x', u') \in \{0,1\}^2$ such that 
$d_{H}(z^{I_{x'u'}},\tilde z^{I_{x'u'}})\geq \frac{1}{4}k_1$. 
%\Comment{Wieso?}
Therefore
\begin{align*}
&N(x'u'|z^n)+d_{H}(z^{I_{x'u'}},\tilde z^{I_{x'u'}})\\
&\geq nW^A(x'|u')/2-k_1W^A(x'|u')/8-\epsilon n+k_1/4\\
&\geq nW^A(x'|u')/2+\frac{7}{64}k_1\;.
\end{align*}
If there exists $(x'',u'')\neq (x',u') \in \{0,1\}^2$ such that
$d_{H}(z^{I_{x''u''}},\tilde z^{I_{x''u''}})\geq \frac{1}{32}k_1$,
then the claim follows from Lemma~\ref{statLemma}. Otherwise
$\tilde z^n \notin \mathcal{T}_{Q,\epsilon}^n$.
\end{proof}

Next, we will prove a technical lemma: 
\begin{lemma}
\label{binom}
 For any $n$ it holds that, if $k \leq np$,
\begin{align*}
\sum_{i=0}^{k}\binom n i p^{i}(1-p)^{n-i} \leq 2^{-2np^2+4pk}.
\end{align*}
\end{lemma}
\begin{proof}
 Let $X_{1},X_{2},\ldots,X_{n}$ be random variables with $Pr[X_{i}=1]=p$ and $Pr[X_{i}=0]=(1-p)$. Let $X=\sum_{i=1}^{n}X_{i}$. Then using Lemma \ref{hoeffding} and setting $t:=np-k$
\begin{align*}
\sum_{i=0}^{k}\binom n i p^{i}(1-p)^{n-i} =\Pr[X \leq k]
&\leq \exp(-2t^2/n)\\
&\leq 2^{-2(np-k)^2/n}\\
&\leq 2^{-2np^2+4pk}.
\end{align*}
\end{proof}

\begin{lemma}
\label{delayLemma}
If Alice does not input any values to at least $k_{2}$ boxes before sending $syn(x^n)$ to Bob, then 
Bob does accept the opening of the protocol with negligible probability.
\end{lemma}

\begin{proof}
Alice does not give any input to at least $k_2$ boxes before sending a syndrome $s_{0}$ to Bob. Later she gives her inputs to the remaining $k_2$ boxes and gets a random output $x_i$ for each box. We know that any two strings $s^n \neq \tilde s^n$ with $syn(s^n)=syn(\tilde s^n)$ have distance at least $d>2k_2$. We can bound the probability that the output string has distance at most $k_1$ to a string with syndrome $s_{0}$ by
\[\sum_{i=0}^{k_{1}}{ \genfrac(){0pt}{}{k_{2}}{i}} p_{0}^i(1-p_{0})^{k_2-i}.\]
Note that since $4p_0+1>1$ and $2p_0 \leq 1$, we have $p_0k_2 \geq k_1$. So we can apply Lemma \ref{binom} and get an upper bound on this probability of
\[2^{-2k_2p_0^2+4k_1p_0} = 2^{-2\frac{k_1 (4p_0 + 1)}{2p_0^2} p_0^2+4k_1p_0} = 2^{-k_1}\;.\]
The statement now follows from Lemma~\ref{statLemma4}.
\end{proof}

\begin{lemma}
\label{bindingLemma}
If Alice changes only $k_{1}$ values and delays only $k_{2}$ inputs, then the protocol is binding.
\end{lemma}
\begin{proof}
Any two input strings $s^n$ and $\tilde s^n$ with $s_{0}=syn(s^n)=syn(\tilde s^n)$ have distance at least $d$. If we ignore all the positions where Alice did not input anything to the box, $s^n$ and $\tilde s^n$ still have distance at least $d-k_{2}>2k_{1}$.   
\end{proof}

\section{Tightness of our Results}

%Next, we show that every binary non-signaling box that is non-local allows to securely implement bit commitment with either Protocol I or Protocol II. 
In this section we show that every non-signaling box with binary outputs that cannot be securely implemented from shared randomness allows to realize bit commitment with one of the above protocols. 
%Let $W$ be a non-signaling box such that there exist $z\neq z'$ with  $\hat W_{z}=\hat W_{z'}$. Then there exists a box $\tilde W$ that is obtained from $W$ by removing an input and . 

\begin{lemma}
\label{lemma:reduced-box}
Let $W:\mathcal{U}\times\mathcal{V}\rightarrow\{0,1\}^2$ be a non-signaling box with $|\mathcal{U}|\geq2$. If there exists $(x_0,u_0)$ such that either $W^A(x_0|u_0)=0$ or $\hat W_{x_0u_0}=\hat W_{x_1u_1}$ for some $(x_1,u_1)\neq (x_0,u_0)$ with  $W^A(x_0|u_0)\leq W^A(x_1|u_1)$, then bit commitment can be implemented from $W$ if and only if bit commitment can be implemented from the reduced box $\tilde W$ that is obtained by removing input $u_0$ from $W$. Furthermore, $W$ is local if and only if $\tilde W$ is local.  
\end{lemma}
\begin{proof}
We prove the statement by showing that Alice having access to $\tilde W$ can simulate the behavior of $W$ on input $u_0$ using local randomness: %$W$ can be securely implemented from $\tilde W$ by simulating the output on input $u_0$ as follows: 
We first consider the case where $\hat W_{x_0u_0}=\hat W_{x_1u_1}$ with $u_1\neq u_0$ and $W^A(x_1|u_1)\neq 0$. We define $p:=W^A(x_0|u_0)/W^A(x_1|u_1)$. Then it holds that $$W(x_0y|u_0v)=pW(x_1y|u_1v)$$ for all $y\in \{0,1\},v\in \mathcal{V}$. It follows from the non-signaling conditions that $$W((1-x_0)y|u_0v)=(1-p)W(x_1y|u_1v)+W((1-x_0)y|u_1v)$$ for all $y\in \{0,1\},v\in \mathcal{V}$. 
We assume $x_0=x_1=0$. Then we can simulate $W$ using $\tilde W$ in the following way: Alice gives input $u_1$ to $\tilde W$ and gets output $x$. If $x=1$, then Alice outputs 1. If $x=0$, then Alice outputs 0 with probability $p$ and 1 with probability $1-p$. If $W^A(x_0|u_0)=0$ or $\hat W_{0u_0}=\hat W_{1u_0}$, then Alice on input $u_0$ outputs 0 with probability $W^A(0|u_0)$ and 1 with probability $W^A(1|u_0)$. 
\end{proof}

\begin{theorem}
\label{completeness1}
A non-signaling box $W: \{0,1\}^2 \rightarrow \{0,1\}^2$ that fulfills neither Condition \ref{property1} nor Condition \ref{non-trivial3} does not allow for information-theoretically secure bit commitment (with noiseless communication only) and is local.
\end{theorem}
\begin{proof}
We first consider the case where there exists $(x_0,u_0)$ such that $W^A(x_0|u_0)=0$ or $\hat W_{x_0u_0}=\hat W_{x_1u_1}$ for some $(x_1,u_1)\neq (x_0,u_0)$. We assume $W^A(x_0|u_0)\leq W^A(x_1|u_1)$ and examine the box $\tilde W$ that is obtained by removing input $u_0$. $\tilde W$ is obviously local. If $\hat W(0|u_1)=\hat W(1|u_1)$, the box is independent and doesn't allow for bit commitment. If there is a perfectly correlated or anti-correlated input pair, the box doesn't allow for bit commitment according to Theorem \ref{impossibility}. Otherwise bit commitment can be reduced to this box using Protocol I. From Lemma \ref{lemma:reduced-box} it follows that we can implement bit commitment from $W$ if and only if bit commitment can be implemented from $\tilde W$.  Thus, the claim follows for all boxes with $W^A(x_0|u_0)=0$ or $\hat W_{x_0u_0}=\hat W_{x_1u_1}$ for some $(x_1,u_1)\neq (x_0,u_0)$. In the following we assume $W^A(x_0|u_0)\neq0$ for all $x_0,u_0\in \{0,1\}$ and $\hat W_{z}\neq\hat W_{z'}$ for all $z,z' \in \{0,1\}^2$ with $z\neq z'$. \\
%We define the set $D:=\{z \in \{0,1\}^2: \hat W_{z}\in \extr(\conv(\hat W))\text{ and } W_{z}\neq \hat W_{\tilde z}~\forall z\neq \tilde z\}$. We distinguish the following cases:\\
(1) $|\extr(\conv(\hat W))|\geq 3$: Then the box fulfills Condition \ref{non-trivial3} and we can securely implement bit commitment using Protocol II.\\ 
(2) $|\extr(\conv(\hat W))|=2$: We first consider the case $\hat W_{1u},\hat W_{0u} \in D$. Without loss of generality, we can assume $u=0$. Then there exist $0<\lambda_0,\mu_0 < 1$ such that
$$\hat W_{01}=\lambda_0\hat W_{00}+(1-\lambda_0) \hat W_{10},$$
$$\hat W_{11}=\mu_0\hat W_{00}+(1-\mu_0) \hat W_{10}.$$
We define $\lambda_1:=1-\lambda_0$ and $\mu_1:=1-\mu_0$. Then it follows from the non-signaling conditions that for all $(y,v)\in \{0,1\}\times\mathcal{V}$
$$\frac{W(0y|1v)}{W^A(0|1)}=\frac{\lambda_0W(0y|0v)}{W^A(0|0)}+\frac{\lambda_1W(1y|0v)}{W^A(1|0)} ,$$
$$\frac{W(1y|1v)}{W^A(1|1)}=\frac{\mu_0W(0y|0v)}{W^A(0|0)}+\frac{\mu_1W(1y|0v)}{W^A(1|0)}.$$
%$$\frac{1}{W^A(0|1)}M_{01}=\frac{\lambda_0}{W^A(0|0)}M_{00}+\frac{\lambda_1}{W^A(1|0)} M_{10},$$
%$$\frac{1}{W^A(1|1)}M_{11}=\frac{\mu_0}{W^A(0|0)}M_{00}+\frac{\mu_1}{W^A(1|0)} M_{10}.$$
We define
%$$a_x:=(\frac{\lambda_x}{W^A(x|0)})W^A(0|1),~x\in\{0,1\},$$
%$$b_x:=(\frac{\mu_x    }{W^A(x|0)})W^A(1|1),~x\in\{0,1\}.$$
$$a_x:=\frac{\lambda_xW^A(0|1)}{W^A(x|0)},~x\in\{0,1\},$$
$$b_x:=\frac{\mu_x  W^A(1|1)  }{W^A(x|0)},~x\in\{0,1\}.$$

Then it follows from the non-signaling conditions that for all $(y,v)\in \{0,1\}\times\mathcal{V}$ it holds that 
$W(0y|0v)+W(1y|0v)$ is equal to 
$$(a_0+b_0)W(0y|0v)+(a_1+b_1)W(1y|0v)$$
%$$M_{00}+M_{10}=(a_0+b_0)M_{00}+(a_1+b_1)M_{10}.$$
As we have have excluded the case $\hat W_{10}=\hat W_{00} $, it follows that 
$a_0+b_0=a_1+b_1=1$.
Then the box is local as follows from $W(xy|uv)=W^A(0|0) V_{A}^0(x|u) V_{B}^0(y|v) + W^A(1|0)V_{A}^1(x|u) V_{B}^1(y|v)$ with
\begin{center}
  \begin{tabular}{ |c|c|c| }
  	\hline
    $(x,u)$&$V_{A}^0(x|u)$&$V_{A}^1(x|u)$ \\\hline\hline
    $(0,0)$&$1$&$0$\\\hline
    $(0,1)$& $a_0$ & $a_1$ \\\hline
    $(1,0)$&$0$&$1$\\\hline
    $(1,1)$&$b_0 $&$b_1$\\\hline
\end{tabular}
\\
\vspace{0.2cm}
\end{center}
and 
$$V_B^0(y|v):=W(0y|0v)/W^A(0|0),$$
$$V_B^1(y|v):=W(1y|0v)/W^A(1|0)$$ for all $y,v \in \{0,1\}$. If one of the inputs $(0,0)$ or $(0,1)$ is perfectly correlated or anti-correlated, then we cannot reduce bit commitment to this box (Theorem~\ref{impossibility}). Otherwise we can securely implement bit commitment from this box using Protocol I. 
\\ Next, we consider the case $\hat W_{x0},\hat W_{x'1} \in D,~x,x'\in\{0,1\}$. We assume $x=x'=0$. Then it holds that
$$\hat W_{10}=\lambda_{00}\hat W_{00}+\lambda_{01}\hat W_{01},$$
$$\hat W_{11}=\mu_{00}\hat W_{00}+\mu_{10}\hat W_{10}.$$
If there is $u\in \{0,1\}$ such that for all $v\in\{0,1\}$ the box is neither perfectly correlated nor perfectly anti-correlated for input $(u,v)$, then the box fulfills Condition \ref{property1}. Otherwise, there must be $v_{0},v_{1}\in\{0,1\}$ such that the box is perfectly correlated or anti-correlated for both $(0,v_{0})$ and $(1,v_{1})$. Then it follows that $\lambda_{00}=0$ and $\mu_{10}=0$, which is a contradiction to our assumptions.\\
The case $|\extr(\conv(\hat W))|\leq 1$ we have already excluded.

\end{proof}

%\section{Proof of Main Theorem}

%\Comment{Was ist der zusammenhang von ``Impossibility'', Condition 1, Condition 2, Condition 3 und completeness? Sollte Hier zuerst erklaert werden...}

In order to prove that we can reduce bit commitment to any box with binary outputs (and general input alphabets  $\mathcal{U}\text{ and } \mathcal{V}$) that cannot be securely implemented from shared randomness we need to give an alternative condition for the security of Protocol II. 
\begin{property}
\label{non-trivial4}
There exist $u_{0},u_{1} \in \mathcal{U},~u_{0}\neq u_{1}$ and $x_{0}, x_{1}\in\{0,1\}$ such that the following two conditions hold:\\
(1) $W_{x_0u_0},W_{x_1u_1}$ are extreme points of $\conv(\hat W)$, i.e., for all $c\in\{(x_0,u_0),(x_1,u_1)\}$ it holds that for all probability distributions $P$ s.t. $P(c)=0$
\begin{align*}
\Big \Vert \hat W_{c}-\sum_{z}P(z) \hat W_{z} \Big \Vert_{1} \geq \delta.
\end{align*}
(2) Let $c,c'\in\{(1-x_{0},u_0),(1-x_1,u_1)\}$ with $c\neq c'$. Then for all probability distributions $P$ such that $P(c')>0$ and $P(c)=0$  it holds that 
$$\Big \Vert \hat W_{c}-\sum_{z}P(z) \hat W_{z} \Big \Vert_{1} \geq \delta.$$
\end{property}
To prove Protocol II secure for all boxes that fulfill Condition \ref{non-trivial4}, we replace Lemma~\ref{statLemma4} with the following lemma. 
We assume that $(x_{0},u_{0})=(0,0)$ and $(x_{1},u_{1})=(0,1)$.  
\begin{lemma}
\label{statLemma5}
If $d_{H}(z^n,\tilde z^n)\geq k_{1}$, then the probability that Bob accepts $\tilde z^n$ is negligible in $n$.
\end{lemma}
\begin{proof}
For all $w^n \in \mathcal{T}_{Q,\epsilon}^n$ it holds that $|N(xu|w^n)-\frac{n}{2}W^A(x|u)|\leq \epsilon n\leq \frac{1}{32}k_1$. We distinguish the following two cases:\\
(1) If there exists $u' \in \{0,1\}$ such that $N(u'|u^n)\leq n/2-k_1/8$, then the statement follows from the proof of Lemma~\ref{statLemma4}.
\\
(2) If $|n/2-N(u|u^n)|\leq k_1/8 \text{ for all }u\in \{0,1\}$, then it follows from Lemma~\ref{typical2} that with overwhelming probability   $z^n\in \mathcal{T}_{W^A,\epsilon}^n(u^n)$. Assume  $z^n\in \mathcal{T}_{W^A,\epsilon}^n(u^n)$. If $d_{H}(z^{I_{00}},\tilde z^{I_{00}})\geq \frac{1}{8}k_1$ or $d_{H}(z^{I_{01}},\tilde z^{I_{01}}) \geq \frac{1}{8}k_1$, then the claim follows from Lemma~\ref{statLemma} and Condition \ref{non-trivial4}. If $|\{i \in I_{10}:z_{i}=11\}|\geq \frac{1}{8}k_1$, then the claim follows from Condition \ref{non-trivial4} and the proof of Lemma~\ref{statLemma} as follows: Let $D:=\{k\in I_{10}:z_k\neq \tilde z_k\}$. We use Condition \ref{non-trivial4} and replace (\ref{statLemmaInequ}) with 
\begin{align*}
\Big \Vert\frac{1}{|I_{10}|}\sum_{k \in I_{10}}W_{z_{k}}-W_{10} \Big \Vert_{1}
&=\frac{|D|}{|I_{10}|} \Big \Vert W_{10}-\frac{1}{|D|}\sum_{k \in D}W_{z_{k}} \Big \Vert_{1}\\
&=\frac{|D|}{|I_{10}|}\delta \geq \frac{1}{8}k_1\delta/n.
\end{align*}
We assume $\tilde z^n \in \mathcal{T}_{Q,\epsilon}^n$. Then it follows as in the proof of Lemma~\ref{statLemma} that $W_{z^n}^n(\mathcal{T}^n_{W,\epsilon}(\tilde z^n))$ is negligible.
The same holds if $|\{i \in I_{11}:z_{i}=10\}|\geq \frac{1}{8}k_1$. In all other cases it follows that $\tilde z^n \notin \mathcal{T}_{Q,\epsilon}^n$.
\end{proof}
\noindent

\begin{comment}
\begin{theorem}
\label{completeness2}
A non-signaling box $W:\mathcal{U}\times\mathcal{V} \rightarrow \{0,1\}^2$ that doesn't allow to securely implement bit commitment with either Protocol I or Protocol II is local and doesn't allow for information-theoretically secure bit commitment at all.
\end{theorem}
\end{comment}

%\Comment{Vielleicht sollte dieses Theorem ganz an den anfang, oder gar in die Einleitung? Es ist ja klar das Main-Theorem.}

\begin{theorem}
\label{completeness2}
Bit Commitment can be reduced to any non-signaling box with binary outputs that cannot be securely implemented from shared randomness.
\end{theorem}

\begin{proof}: 
%In case $\left|\mathcal{U}\right|\leq 2$ the statement follows from the proof of Theorem \ref{completeness1}. If $\left|\mathcal{U}\right|>2$, 
%The proof proceeds in the same way as the proof of Theorem \ref{completeness1}: We eliminate the cases where there 
If $|\mathcal{U}|\leq 2$, then the statement follows from the proof of Theorem \ref{completeness1}. Otherwise, we first eliminate the cases where there exists $(x_0,u_0)$ such that $W^A(x_0|u_0)=0$ or $\hat W_{x_0u_0}=\hat W_{x_1u_1}$ for some $(x_1,u_1)\neq (x_0,u_0)$ by using Lemma \ref{lemma:reduced-box} to reduce the box. Then we consider $D:=\extr(\conv(\hat W))$: In the case $|D|=2$ the statement is proven in the same way as in the proof of Theorem \ref{completeness1}. The case $|D|\geq3$ is a little bit more involved: If there is $\hat W_{1u},\hat W_{0u} \in D$, then Condition $\ref{non-trivial3}$ is fulfilled and we can implement bit commitment using Protocol II. Otherwise, we can either implement bit commitment using Protocol I or for every input $u$ corresponding to an element of $D$ there is an input $v$ for Bob such that the box is perfectly correlated or anti-correlated. Let $\hat W_{x_0u_0}\in D$. Without loss of generality we assume that $W$ is perfectly correlated for input $(u_0,v_0)$. Then there exist $\lambda_z$ with $\sum_{z:\hat W_z\in D}\lambda_z=1$ such that
$$\hat W_{(1-x_0)u_0}=\sum_{z:\hat W_z\in D}\lambda_zW_z.$$
%As $W$ is perfectly correlated for $(0,0)$, we know that $\lambda_{00}=0$. Then t
There exists $(x_1,u_1)$ with $u_1\neq u_0$ such that $\lambda_{x_1u_1}>0$. We assume $x_0=x_1=0$. We have $W(10|u_0v_0)=0$. This implies $W(00|u_1v_0)=0$. From the non-signaling conditions follows that $W(10|u_1v_0)=W(00|u_0v_0)>0$. There exists $v_1\in \mathcal{V}$ such that $(u_1,v_1)$ is perfectly correlated or anti-correlated. We assume without loss of generality that $(u_1,v_1)$ is perfectly correlated. This implies $W(00|u_1v_1)>0$ and $W(10|u_1v_1)=0$. From $\lambda_{x_1u_1}>0$ follows that $W(10|u_0v_1)>0$. So we have $\hat W_{0u_0},\hat W_{0u_1}\in D$, $W(10|u_0v_0)=W(10|u_1v_1)=0$, $W(10|u_1v_0)>0$ and $W(10|u_0v_1)>0$. Thus, Condition \ref{non-trivial4} is fulfilled.

\end{proof}

\section{Concluding Remarks}

We have shown that any bipartite non-signaling system with binary outputs can either be securely realized from shared randomness or allows for bit commitment. %This all-or-nothing result implies, in particular, that the classical measurement-outcome behavior of a two-qubit state can be used for bit commitment if it has no hidden-variable explanation. 

An obvious open question is whether a similar result holds for arbitrary output alphabets. Furthermore, it would be interesting to know whether oblivious transfer can be implemented from the same set of non-signaling systems. %what circumstances (the stronger functionality of) oblivious transfer can be obtained. In certain settings, e.g., distributed information or noisy channels, bit commitment and oblivious transfer have turned out to be  realizable from exactly the same starting points. 

\section*{Acknowledgments}
We thank Dejan Dukaric, Esther H\"anggi and Thomas Holenstein for helpful discussions, and the referees for their useful comments.

%\bibliographystyle{IEEEtran}
%\bibliography{all}
\bibliography{all}
\bibliographystyle{IEEEtran}

\end{document}